\documentclass[twoside,journal]{IEEEtran}
\usepackage{amsfonts}
\usepackage{amssymb}
\usepackage{amsmath}
\usepackage{mathrsfs}
\usepackage{verbatim}
\usepackage{subfigure}
\usepackage{balance}
\usepackage{booktabs}
\usepackage{fancybox}
\usepackage{bm}
\usepackage{changepage}
\usepackage{extarrows}
\usepackage{algorithm}
\usepackage{algorithmic}
\usepackage{diagbox}
\usepackage{multirow}
\usepackage{array}
\usepackage{graphicx}
\usepackage{epstopdf}
\usepackage{caption}
\usepackage{color}
\newtheorem{lemma}{Lemma}
\newtheorem{corollary}{Corollary}
\newtheorem{proof}{Proof}
\newtheorem{proposition}{Proposition}

\begin{document}

\title{Physical Layer Security in Large-Scale Random Multiple Access Wireless Sensor Networks: A Stochastic Geometry Approach}
\author{Tong-Xing Zheng, \IEEEmembership{Member, IEEE,}
Xin Chen, Chao Wang, \IEEEmembership{Member, IEEE}, Kai-Kit Wong, \IEEEmembership{Fellow, IEEE}, and Jinhong Yuan, \IEEEmembership{Fellow, IEEE}

\thanks{Tong-Xing Zheng and Xin Chen are with the School of Information and Communications Engineering, Xi'an Jiaotong University, Xi'an 710049, China, also with the National Mobile Communications Research Laboratory, Southeast University, Nanjing 210096, China, and also with the Ministry of Education Key Laboratory for Intelligent Networks and Network Security, Xi'an Jiaotong University, Xi'an 710049, China (e-mail: zhengtx@mail.xjtu.edu.cn, cx0513@stu.xjtu.edu.cn).}
\thanks{Chao Wang is with the State Key Laboratory of Integrated Services Networks, Xidian University, Xi'an 710071, China, and also with the National Mobile Communications Research Laboratory, Southeast University, Nanjing 210096, China (e-mail: drchaowang@126.com).}%
\thanks{Kai-Kit Wong is with the Department of Electronic and Electrical Engineering, University College London, WC1E 6BT London, U.K. (e-mail: kai-kit.wong@ucl.ac.uk).}	%
\thanks{Jinhong Yuan is with the School of Electrical Engineering and	Telecommunications, University of New South Wales, Sydney, NSW 2052, Australia (e-mail: j.yuan@unsw.edu.au).}%
}
	\maketitle
	\vspace{-0.8 cm}
	
\begin{abstract}
This paper investigates physical layer security for a large-scale WSN with random multiple access, where each fusion center in the network randomly schedules a number of sensors to upload their sensed data subject to the overhearing of randomly distributed eavesdroppers. 
We propose an uncoordinated random jamming scheme in which those unscheduled sensors send jamming signals with a certain probability to defeat the eavesdroppers.
With the aid of stochastic geometry theory and order statistics, we derive analytical expressions for the connection outage probability and secrecy outage probability to characterize transmission reliability and secrecy, respectively.
Based on the obtained analytical results, we formulate an optimization problem for maximizing the sum secrecy throughput subject to both reliability and secrecy constraints, considering a joint design of the wiretap code rates for each scheduled sensor and the jamming probability for the unscheduled sensors.  
We provide both optimal and low-complexity sub-optimal algorithms to tackle the above problem, and further reveal various properties on the optimal parameters which are useful to guide practical designs.
In particular, we demonstrate that the proposed random jamming scheme is beneficial for improving the sum secrecy throughput, and the optimal jamming probability is the result of trade-off between secrecy and throughput. We also show that the throughput performance of the sub-optimal scheme approaches that of the optimal one when facing a stringent reliability constraint or a loose secrecy constraint. 
\end{abstract}

	
\IEEEpeerreviewmaketitle
	
\section{Introduction}
\IEEEPARstart{W}{ireless} sensor networks (WSNs) have drawn  prominent research interests from both academia and industry in recent years and have been envisioned as key technologies for Internet-of-Things (IoTs) \cite{Razzaque2016Middleware}. 
With a massive number of sensors deployed in a network, collecting and reporting diverse environmental data to fusion centers (FCs), WSNs show tremendous potential in a variety of applications, including security and battlefield surveillance, disaster alert, industrial automation, traffic management, smart healthcare and homes, etc \cite{Vo20185G}. 
However, delivering sensing data over the air is prone to eavesdropping attacks due to the openness of wireless channels. 
Moreover, it is challenging to employ key-based cryptographic techniques for WSNs, where the distribution, maintenance, and management of secret keys are expensive under dynamic and large-scale topologies.
In general, sensors are incapable of implementing complicated cryptographic algorithms due to low energy and computing power \cite{Chen2009Sensor}.
In this context, \emph{physical layer security} \cite{Le2021Physical}--\cite{Yang2015Safeguarding} has emerged as an appealing low-complexity paradigm to realize secure transmissions by exploiting wireless media characteristics, and it promises to be a powerful supplement or even  alternative to the cryptographic mechanisms for WSNs. 

\subsection{Background and Motivations}
The research of physical layer security dates back to as early as 1975 when Wyner postulated the information-theoretic foundation in his ground-breaking treatise \cite{Wyner1975The}.  
Wyner introduced the degraded witetap channel model and showed that the reliability and secrecy of information delivery can be concurrently guaranteed with appropriate secrecy channel coding. 
This pioneering work has motivated substantial endeavors invested  in developing physical layer security during the past decade, from the viewpoints of both information theory and signal processing.

Early research on physical layer security have been more concentrated on point-to-point communication links. 
Fundamental information-theoretic limits and optimal secrecy signaling schemes have been investigated by exploiting the channel state information (CSI) of both the main channel (spanning from transmitter to the intended receiver)  and the wiretap channel (spanning from transmitter to the undesired receiver, or eavesdropper). 
When the eavesdropper's CSI is completely unavailable, Goel and Negi \cite{Goel2008Guaranteeing} proposed to radiate controllable artificial noise or jamming signals along with confidential information, through either centralized multiple antennas or distributed cooperative jammers, to degrade the wiretap channel while without impairing the main channel. 
With no need for the eavesdropper's CSI, the idea of artificial noise or cooperative jamming has opened a new avenue for enhancing physical layer security and has sparked a wave of innovation, e.g, see \cite{Zheng2015Multi}--\cite{Wang2020Energy}.

Different from the point-to-point scenarios, secure  communications in large-scale wireless networks suffer from severe interference caused by a large amount of concurrent transmissions, and therefore the security performance depends heavily on the network geometry and the locations of nodes in the network. 
Against this background, stochastic geometry theory has offered powerful tools to study large-scale wireless networks from a statistical point of view by modeling node positions as some spatial distributions like Poisson point process (PPP) \cite{Haenggi2009Stochastic}, and the research on physical layer security under a stochastic geometry framework has been extensively carried out recently.
For example, for large-scale ad hoc networks, Zhou \emph{et al.} \cite{Zhou2011Throughput} and Zhang \emph{et al.} \cite{Zhang2013Enhancing} respectively explored single- and multi-antenna secure transmissions and identified the tradeoff between reliability and secrecy against eavesdropping attacks. 
Zheng \emph{et al.} \cite{Zheng2017Safeguarding,Zheng2017Physical} explored the great benefit of full-duplex receiver jamming in enhancing the network-wide secrecy throughput and energy efficiency.
For multi-cell cellular networks, Wang \emph{et al.} \cite{Wang2013Physical} investigated the secure downlink transmissions and discussed the impact of cell association and the location information of mobile users. Geraci \emph{et al.} \cite{Geraci2014Physical} further evaluated the achievable secrecy rate with regularized channel inversion precoding under a massive multiple-input multiple-output (MIMO) system. Wang \emph{et al.} \cite{Wang2016Physical} comprehensively analyzed the network-wide secrecy for a multi-tier heterogeneous cellular network, where  a threshold-based mobile association policy was proposed to balance link quality and secrecy. 
Wang \emph{et al.} \cite{Wang2020Physical} further applied the artificial noise aided physical layer security to the cellular vehicle-to-everything (C-V2X) networks.
Interested readers are referred to \cite{Wang2016Physical_book} for a more thorough understanding of the physical layer security in random wireless networks under the stochastic geometry framework.

As mentioned previously, physical layer security is particularly important for WSNs, since employing traditional cryptographic mechanism is rather costly and difficult. Recently, physical layer security has been advocated to protect communications from eavesdropping for WSNs.
The majority of existing literature on physical layer security in WSNs has been concentrated on deterministic network geometry, i.e., ignoring the uncertainty of nodes' locations or large-scale path loss \cite{Marano2009Distributed}--\cite{Barcelo-Llado2014Amplify}.
Given that sensors are generally randomly scattered, Lee \emph{et al.} \cite{Lee2013Distributed} first introduced the concept of distributed network secrecy and quantified the secrecy throughput and energy consumption for a multilevel WSN using tools from stochastic geometry. 
Deng \emph{et al.} \cite{Deng2016Physical} further analyzed the average secrecy rate for a three-tier WSN.
However, these works only considered access technologies with orthogonal resource blocks (RBs). 

Random multiple access has the virtue of being highly convenient and flexible without requiring a complicated control scheduling, which is well-suited for the large-scale WSNs particularly when the system load is overly heavy.
Although studies of secure multiple access have been reported for various wiretap channel models, the results cannot be directly applied for large-scale WSNs with stochastic network geometry.
In recent years, the physical layer security of non-orthogonal multiple access (NOMA) for large-scale networks has received considerable attention, e.g.,  \cite{Liu2017Enhancing}--\cite{Zhang2018Enhancing}, but unfortunately, at present researchers have mainly focused on two-user pairing sharing the same RB.
In fact, random multiple access with non-orthogonal RBs will significantly hamper the analysis of channel statistics for large-scale WSNs, since we have to deal with the combined effect of channel fading, the random locations of external interfering sensors, as well as the uncertainty of the successive interference cancellation (SIC) based decoding order for internal sensors belong to the same FC. 
The intractability of analysis will in no doubt make it challenging to design schemes to optimize the network security performance in terms of e.g., sum secrecy throughput. Our research work aims to provide an analytical framework and design schemes to address the aforementioned problem.

\subsection{Our Work and Contributions}
In this paper, we study physical layer security for a large-scale WSN consisting of randomly deployed sensors and FCs, coexisting with randomly distributed eavesdroppers attempting to intercept the data broadcast by the sensors. 
We establish a joint analysis and design framework to evaluate the transmission reliability and secrecy and optimize the network-wide performance in terms of the sum secrecy throughput.
Our main contributions care summarized as follows:

\begin{itemize}
\item 
We propose a random multiple access strategy which associates each sensor to its nearest FC, and each FC randomly selects a certain number of sensors for data acquisition. 
We then propose an uncoordinated jamming scheme  to combat eavesdropping where those unscheduled sensors, who are not chosen for data collection at the current time slot, independently radiate jamming signals with a certain probability.
\item
We assume that each FC adopts  zero-forcing SIC (ZF-SIC) to decode the multiple streams of the scheduled sensors, where the decoding order is determined according to their distances to the associated FC.
We derive new closed-form expressions for the connection outage probability of a typical FC, leveraging tools from the stochastic geometry theory and order statistics.  We also provide analytical expressions for the secrecy outage probability of the typical FC, assuming that eavesdroppers employ the minimum mean square error (MMSE) receiver to demodulate signals and have a powerful multi-user detection capability.
\item
We formulate a problem of maximizing the sum secrecy throughput of the typical FC, imposing both reliability and secrecy constraints on each scheduled sensor. We jointly design the optimal parameters, including the code rates of the scheduled sensors and the jamming probability of the unscheduled sensors.
We also provide a computational-convenient sub-optimal solution by forcing each scheduled sensor to attain a target high level of reliability. We derive closed-form expressions for the optimal code rates, and we prove that the sum secrecy throughput is a quasi-concave function of the jamming probability, where the optimal jamming probability can be efficiently calculated via the bisection method.  	
\end{itemize}

\subsection{Organization and Notations}
The remainder of this paper is organized as follows. 
Section II describes the random multiple access WSN and the optimization problem of interest. 
Section III analyzes the connection and secrecy outage probabilities of the secure transmission of the scheduled sensors. Section IV details the sum secrecy throughput maximization, with both optimal and sub-optimal solutions provided. Section V concludes this paper.

\emph{Notations}: Bold uppercase (lowercase) letters denote matrices (column vectors). $|\cdot|$, $\|\cdot\|$, $(\cdot)^{\dagger}$, $(\cdot)^{\rm T}$, $\ln(\cdot)$, $\mathbb{P}\{\cdot\}$, and $\mathbb{E}_z[\cdot]$ denote the absolute value, Euclidean norm, conjugate, transpose, natural logarithm, probability, and the expectation over a random variable $z$, respectively.
$f_z(\cdot)$ and $\mathcal{F}_z(\cdot)$ denote the probability density function (PDF) and cumulative distribution function (CDF) of $z$, respectively. 
$\mathcal{CN}(\mu,v)$, ${\rm Exp}(\lambda)$, and ${\rm Gamma}(N,\lambda)$ denote the circularly symmetric complex Gaussian distribution with mean $\mu$ and variance $v$, the exponential distribution with parameter $\lambda$, and the gamma distribution with parameters $N$ and $\lambda$, respectively. $\mathbb{R}^{m\times n}$ and $\mathbb{C}^{m\times n}$ denote the $m\times n$ real and complex number
domains, respectively.
In addition, $\binom{n}{m}\triangleq\frac{n!}{m!(n-m)!}$ for integers $n>m\ge0$.

\section{Network Model and Problem Description}
We consider the issue of secure wireless transmissions for a large-scale WSN as illustrated in Fig. \ref{figSys}, where a large number of sensor nodes continually monitor the surrounding environment and report their observations to the FCs which are responsible for decision making, whilst the ongoing data uploading is overheard by eavesdroppers hiding in the network. We assume that the sensors, FCs, and eavesdroppers are all spatially randomly positioned, and their locations are modeled as independent homogeneous PPPs $\Phi_s$, $\Phi_c$, and $\Phi_e$ in a two-dimensional plane $\mathbb{R}^2$, with spatial densities $\lambda_s$, $\lambda_c$, and $\lambda_e$, respectively.\footnote{Throughout this paper, we have a slight abuse with the notation $\Phi$, which is used to represent the set of nodes' locations as well as the nodes themselves.}

\begin{figure}[!ht]
	\centering
	\includegraphics[width=3.5in]{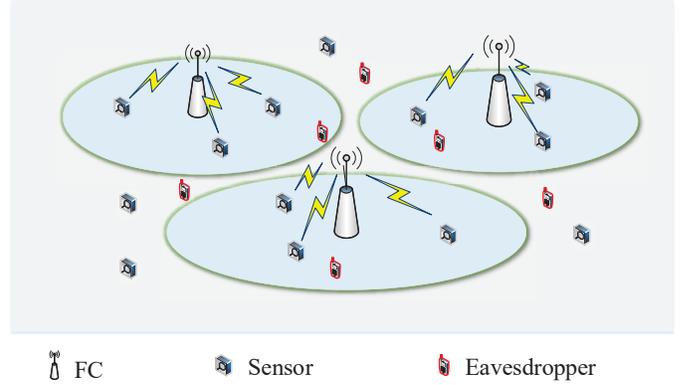}
	\caption{Illustration of a large-scale security-oriented WSN. A great quantity of FCs are deployed in the network (three FCs in the figure), each of which collects different categories of environmental information from a certain number (three sensor nodes as a group for the same FC in the figure) of near sensors nodes in the presence of numerous randomly located eavesdroppers (six eavesdroppers in the figure).}
	\label{figSys}
\end{figure}

\subsection{Channel Model}
We consider that the sensors are single-antenna devices due to hardware restrictions, and the FCs and eavesdroppers are equipped with $M_c$ and $M_e$ antennas, respectively, for achieving signal enhancement, interference suppression, etc. 
All the wireless channels are modeled by the combination of a frequency flat Rayleigh fading and a standard distance-based path loss. 
Hence, the channels from a sensor located at $x$ to an FC located at $y$ and to an eavesdropper located at $z$ are respectively characterized as $\bm h_{y,x} r_{y,x}^{-\alpha/2}$ and $\bm g_{z,x} r_{z,x}^{-\alpha/2}$, respectively. To be specific, $\bm h_{y,x}\in\mathbb{C}^{M_c\times 1}$ and $\bm g_{z,x}\in\mathbb{C}^{M_e\times 1}$ represent the small-scale fading channel vector with independent and identically distributed (i.i.d.) entries obeying the distribution $\mathcal{CN}(0,1)$, $r_{y,x}$ and $r_{z,x}$ denote the corresponding Euclidean distances, and $\alpha>2$ is the path-loss exponent.  

\subsection{Random Multiple Access}
We consider random multiple access scheduling, where each sensor is associated with its nearest FC, and each FC randomly chooses a set of $K<M_c$ sensors for data gathering at a given radio resource. Define $\mathcal{K}\triangleq\{1,2,\cdots,K\}$. Note that due to the ultra-dense deployment of sensors, we consider a plausible scenario where the density of sensors is much higher than $K$ times the density of FCs, i.e., $\lambda_s\gg K\lambda_c$, and there always exist more than $K$ sensors assigned to the same FC. 
Under this circumstance, all the sensors in a specific time slot can be naturally divided into two thinned PPPs, namely, the scheduled sensors $\Phi_{a}$ with density $\lambda_a=K\lambda_c$ which are communicating with their associated FCs and the unscheduled sensors $\Phi_{i}$ with density $\lambda_i=\lambda_s - K\lambda_c$ which remain silent, respectively.

\subsection{Uncoordinated Random Jamming}
We assume that each FC knows perfectly the instantaneous CSI regarding its $K$ scheduled sensors, whilst only has the statistical CSI of the other sensors and of the eavesdroppers.\footnote{Theoretically, an FC can obtain the perfect CSI of its scheduled sensors via channel estimation by letting them transmit orthogonal training sequences simultaneously.} 
In order to combat eavesdropping effectively while avoiding bringing severe interference to the FCs, an uncoordinated random jamming scheme is proposed, in which the unscheduled sensors radiate jamming signals at a probability $\rho\in[0,1]$.
By doing this, the distribution of the jamming sensors follows a PPP $\Phi_j$ with density $\lambda_j = \rho\lambda_i$.
 
	We emphasize that the proposed random jamming scheme is suitable for the energy-limited sensor networks owing to its low-level collaboration. This is fundamentally different from those higher-level collaboration schemes such as coordinated ZF jamming, which will cause high overhead and implementation complexity due to information sharing, beamformer design, and synchronization.
	Moreover, the jamming probability $\rho$ is carefully designed off line for maximizing the network security performance, as will be discussed in Sec. IV, and hence our scheme can balance well between network performance and complexity.

\subsection{Multi-Stream Decoding}
At the FC side, the ZF-SIC method is employed to separate the multiple data streams received from its associated $K$ sensors.\footnote{ZF is a typical linear filter for multi-user communication systems, and ZF-SIC is commonly used in an NOMA system to achieve SIC due to its ease of implementation and low computational complexity \cite{Jiang2017ZF,Jiang2018SIC}. In this sense, ZF-SIC is more subtable for the large-scale random multiple access WSN compared with more advanced but complicated methods, e.g., MMSE-SIC.} 
Theoretically, the SIC order should be sorted according to the instantaneous received signal strengths from the strongest to the weakest.
However, given that the impact of large-scale path loss is generally more dominant on the channel impairment and is more stable compared with the small-scale channel fading, we schedule the SIC order based on the sensors' distances to the FC from the nearest to the farthest. 
To be more specific, the procedure of ZF-SIC can be described as follows: 1) first decode the signal received from the nearest sensor by removing the aggregated signals received from the $K-1$ farther sensors through projecting these signals on to the null space of the instantaneous channel of the nearest sensor, 2) then cancel the decoded signal from the composite received signals, and 3) successively decode the signal from the second nearest sensor in a similar way, and so on. 
After the ZF-SIC operation, when decoding the signal from any specific sensor, the FC can successfully eliminate the interfering signals generated by the other $K-1$ sensors.

\subsection{Performance Metrics and Optimization Problem}
In order to secure the data transmission, Wyner's secrecy-preserving channel code, generally known as the wiretap code, is employed. 
In the wiretap code, the rates of the transmitted codewords and the embedded confidential messages are represented by the codeword rate $R_t$ and the secrecy rate $R_s$, respectively. 
The rate redundancy $R_e \triangleq R_t - R_s$ is intentionally introduced for guaranteeing secrecy against eavesdropping attacks. 
If the capacity of the main channel falls below the codeword rate $R_t$, the desired receiver cannot recover the codeword correctly, which is regarded as connection outage, and the probability that this event happens is termed \emph{connection outage probability} (COP). 
If the capacity of the wiretap channel exceeds the rate redundancy $R_e$, perfect secrecy is not possible, which is considered to be secrecy outage, and the probability of this event occurring is referred to as \emph{secrecy outage probability} (SOP). 

Without loss of generality, we focus on a typical FC which is placed at the origin $o$ of the polar coordinate, and denote its $K$ scheduled sensors as $S_1,S_2,\cdots,S_K$ with an ascending sort order of their distances. 
The codeword rate and the secrecy rate of sensor $S_k$ are denoted as $R_{t,k}$ and $R_{s,k}$, respectively, and the corresponding rate redundancy is given by $R_{e,k} = R_{t,k}- R_{s,k}$.
The COP of $S_k$ can be defined as 
\begin{equation}\label{cop_def}
p_{co,k} \triangleq \mathbb{P}\{{\rm SINR}_{o,k}<\beta_{t,k}\},  \quad \forall k\in\mathcal{K},
\end{equation}
where ${\rm SINR}_{o,k}$ denotes the instantaneous signal-to-interference-plus-noise ratio (SINR) of the typical FC for resolving the signal from $S_k$, and $\beta_{t,k} \triangleq 2^{R_{t,k}}-1$ is the threshold SINR for connection outage. 

We consider the wiretap scenario in which the eavesdroppers do not collude with each other and decode messages individually.
In this case, a secrecy outage event takes place if only confidential information is leaked to the most threatening eavesdropper of the highest SINR.
Therefore, the SOP of $S_k$ can be defined as 
\begin{equation}\label{sop_def}
p_{so,k} \triangleq \mathbb{P}\left\{\max_{e\in\Phi_e}{\rm SINR}_{e,k}>\beta_{e,k}\right\},  \quad \forall k\in\mathcal{K},
\end{equation}
where ${\rm SINR}_{e,k}$ denotes the instantaneous SINR of the eavesdropper located at $e\in\Phi_e$ for intercepting the data from $S_k$, and $\beta_{e,k} \triangleq 2^{R_{e,k}}-1$ is the threshold SINR for secrecy outage. 

This paper uses the metric \emph{sum secrecy
throughput} to assess the capacity of  multi-access secure transmissions from the viewpoint of secrecy outage, where the sum secrecy throughput of a typical FC is defined as the total average successfully received confidential information bits from its $K$ scheduled sensors per second per Hertz per channel use subject to certain secrecy constraints.
Formally, the sum secrecy throughput can be formulated as 
\begin{equation}\label{sst_def}
\mathcal{T} = \sum_{k=1}^K R_{s,k}\left(1-p_{co,k}\right), ~ {\rm s.t.}~ p_{so,k}\le \epsilon, ~ \forall k\in\mathcal{K},
\end{equation}
where $\epsilon\in[0,1]$ is a prescribed threshold representing the maximal tolerable SOP.

In this paper, we aim to maximize the sum secrecy throughput $\mathcal{T} $ via jointly designing the wiretap code rates for each scheduled sensor (including the codeword rate $R_{t,k}$ and the secrecy rate $R_{s,k}$) and the jamming probability $\rho$ for the unscheduled sensors.
Before proceeding to the optimization problem, we will first derive analytical expressions for the COP $p_{co,k}$ and SOP $p_{so,k}$ for $k\in\mathcal{K}$ in the following section.

\section{Analyses of COP and SOP}
This section evaluates the reliability and secrecy performance of the large-scale WSN under investigation with random multiple access. Specifically, we will analyze in detail the COP $p_{co,k}$ and SOP $p_{so,k}$ of the secure transmission from the $k$-th nearest sensor $S_k$ to the typical FC located at the origin $o$, utilizing the stochastic geometry theory and order statistics.
For ease of notation, we define $\delta\triangleq{2}/{\alpha}$ and $\phi \triangleq \pi \Gamma(1+\delta) \Gamma(1-\delta)$ throughout the paper.

\subsection{General Results for COP}
Denote the locations of the $K$ sensors $S_k$ scheduled by the typical FC as $s_k$ for $k\in\mathcal{K}$, with an ascending order of their distances to the FC $L_{1}\leq L_{2},\cdots,\le L_{K}$.
Note that due to random multiple access and SIC order scheduling, the ordered distance $L_k$ is a random variable, the statistics of which is characterized by the following lemma. 
\begin{lemma}\label{lemma_lk}
	The PDF of the ordered distance $L_k$ from the typical  FC to the $k$-th nearest sensor is given by
	\begin{equation}\label{lk_pdf}
	f_{L_{k}}(r)=2k \binom{K}{k}\sum_{l=0}^{k-1}\binom{k-1}{l}(-1)^l\pi \lambda_cr e^{-\pi\lambda_cr^2(K-k+l+1)},
	\end{equation}
\end{lemma}
\begin{proof}
	The PDF of $L_k$ follows from  order statistics \cite{David2003Order},
	\begin{equation}
	f_{L_{k}}(r)=k\binom{K}{k} \mathcal{F}_{L}(r)^{k-1}\left[1-\mathcal{F}_{L}(r)\right]^{K-k} f_{L}(r),
	\end{equation}
	where $f_{L}(r)=2\pi r \lambda_ce^{-\pi\lambda_cr^2}$ and $\mathcal{F}_{L}(r)=1-e^{-\pi\lambda_cr^2}$ are the PDF and CDF of the unordered distance $L$ from a sensor to its nearest FC, respectively.
\end{proof}

According to the ZF-SIC decoding described in Sec. II-B, the instantaneous SINR of $S_k$ can be formulated as
\begin{equation}\label{sinr_ok}
{\rm SINR}_{o,k} = \frac{P_a|\bm w_{k}^{\rm T} \bm h_{o,s_k}|^2L_{k}^{-\alpha}}
{I_a+I_j+\omega},
\end{equation}
where $I_a=\sum_{x\in\Phi_a\setminus o}P_a|\bm w_{k}^{\rm T}\bm h_{o,x}|^2r_{o,x}^{-\alpha}$ denotes the interference generated by those sensors scheduled by the FCs other than the typical FC, 
$I_j = \sum_{y\in\Phi_j} P_j|\bm w_{k}^{\rm T}\bm h_{o,y}|^2r_{o,y}^{-\alpha}$ denotes the power of the aggregated jamming signal from the unscheduled sensors, with $P_a$ and $P_j$ being the transmit power of the information-bearing signals and the jamming signals, respectively, and $\omega$ denotes the power of the receiver noise.
Here, $\bm w_{k} = \frac{\bm U_k^{\dagger}\bm U_k^{\rm T}\bm h_{o,s_k}^{\dagger}}{\|\bm U_k^{\rm T}\bm h_{o,s_k}^{\dagger}\|}$ denotes the weight vector designed for the $k$-th sensor as per the ZF-MRC criterion, where $\bm U_k\in\mathbb{C}^{M_c\times M_k}$, with $M_k \triangleq M_c - K + k$, is the projection matrix onto the null space of the matrix $[\bm h_{{o,s_{k+1}}}, \cdots,\bm h_{o,s_K}]$ such that $\bm w_{k}^{\rm T}\bm h_{o,s_j}=0$ for $j>k$. Note that the columns of $\bm U_k$ constitute an orthogonal basis, and hence $|\bm w_{k}^{\rm T} \bm h_{o,s_k}|^2= \|\bm U_k^{\rm T}\bm h_{o,s_k}^{\dagger}\|^2$ and $|\bm w_{k}^{\rm T}\bm h_{o,x}|^2$ obey the gamma distribution ${\rm{Gamma}}(M_k,1)$ and the exponential distribution ${\rm Exp}(1)$, respectively.

The COP of sensor $S_k$ is defined in \eqref{cop_def} with ${\rm SINR}_{o,k}$ given in \eqref{sinr_ok}.
Note that the COP is affected by various uncertainties, including fading channels, node locations, as well as the decoding order.
In the following proposition, we provide an expression for the exact COP. 
\begin{proposition}\label{proposition_cop_general}
	The COP of the secure transmission from the $k$-th nearest sensor $S_k$ to the typical FC is given by
	\begin{align}\label{cop_general}
&	p_{co,k}=  1 -\pi\lambda_ck\binom{K}{k}\sum_{l=0}^{k-1}\binom{k-1}{l}\sum_{m=0}^{M_k-1}\sum_{p=0}^m\binom{m}{p}	\frac{(-1)^{l}}{m!}\times	\nonumber\\
	& 
	\left(\frac{\omega\beta_{t,k}}{P_a}\right)^{m-p}\left[\bm 1_{p=0}\Omega_{\frac{m\alpha}{2}}+\bm 1_{p\neq0}\sum_{n=1}^p\left(\delta\phi\lambda_o\beta_{t,k}^{\delta}\right)^n\Omega_{\mu}\Upsilon_{p,n}\right]	
	\end{align}
	where $\bm 1_{\mathcal{H}}$ is the indicator function with $\bm 1_{\mathcal{H}}=1$ when event $\mathcal{H}$ is true and $\bm 1_{\mathcal{H}}=0$ otherwise, $\lambda_o \triangleq \lambda_a+\left({P_j}/{P_a}\right)^\delta \lambda_j$,	$\mu = \frac{\alpha}{2}(m-p)+n$, $\Omega_{\mu}\triangleq \int_0^{\infty}x^{\mu}e^{-\tau_1x^{\alpha/2}-\tau_2x}dx$ with $\tau_1 = \omega\beta_{t,k}/P_a$ and $\tau_2=\phi\lambda_o\beta_{t,k}^{\delta}+\pi\lambda_c(K-k+l+1)$, and $\Upsilon_{p,n}\triangleq\sum_{\psi_j\in \mathrm{comb}\binom{p-1}{p-n}}\prod_
	{		\substack{
			q_{ij}\in\psi_j\\
			i=1,\cdots,p-n}	}
	\left[q_{ij}-\delta(q_{ij}-i+1)\right]$, with the convention that $\Upsilon_{p,p}=1$ for $p\ge1$.
The term $\mathrm{comb}\binom{p-1}{p-n}$ represents the set of all distinct subsets of the natural numbers $\{1,2,\cdots,p-1\}$ with cardinality $p-n$.
	The elements in each subset $\psi_j$  are sorted in an ascending order with $q_{ij}$ being the $i$-th element of $\psi_j$.
\end{proposition} 
\begin{proof}
	Please refer to Appendix \ref{proof_proposition_cop_general}.
\end{proof}

Although \eqref{cop_general} seems difficult to analyze due to the existence of the integral term $\Omega_{\mu}$, 
it provides a general and accurate expression for the COP without requiring time-consuming simulations. More importantly, it can be used as a baseline for comparison with other approximate results. 
For a special case with $\alpha=4$, $\Omega_{\mu}$ can be simplified by \cite[Eq. (3.462.1)]{Gradshteyn2007Table} as the following practically closed form,
\begin{equation}\label{omega}
\Omega^{\alpha=4}_{\mu} = (2\tau_1)^{-\frac{\mu+1}{2}}\Gamma(\mu+1)\exp\left(\frac{\tau_2^2}{8\tau_1}\right)\mathcal{D}_{-\mu-1}\left(\frac{\tau_2}{\sqrt{2\mu_1}}\right),
\end{equation}
where $\mathcal{D}_{-\mu}(z)$ denotes the parabolic cylinder function 	\cite[Eq. (9.241.2)]{Gradshteyn2007Table}. 
Note that with \eqref{omega}, the new expression of the COP $p_{co,k}$ becomes rather computationally convenient which requires only the calculation or lookup of a $\mathcal{D}_{-\mu}(z)$ value.

\subsection{Interference-Limited Case for COP}
Owing to a large amount of uncoordinated concurrent transmissions in the network, the aggregate interference at a receiver generally dominates the thermal noise. 
Motivated by this fact, we turn to examine the interference-limited WSN by ignoring the receiver noise at the FC side.

The following corollary provides a closed-form expression for the COP $p_{co,k}$ considering the interference-limited case.
\begin{corollary}\label{corollary_cop_int}
	For the interference-limited WSN, the COP of $S_k$ is given by
	\begin{align}\label{cop_int}
	p_{co,k} = &1 - \pi\lambda_ck\binom{K}{k}\sum_{l=0}^{k-1}\binom{k-1}{l}\frac{(-1)^{l}}{\tau_2}\times\nonumber\\
	&
	\left(1+\sum_{m=1}^{M_k-1}\sum_{n=1}^{m}
	\frac{n!}{m!}\left(\frac{\delta\phi\lambda_o\beta_{t,k}^{\delta}}{\tau_2}\right)^n\Upsilon_{m,n}\right).
	\end{align}
\end{corollary}
\begin{proof}
	Please refer to Appendix \ref{proof_corollary_cop_int}.
\end{proof}

It should be noted that although \eqref{cop_int} is in a closed form, the Diophantus equation therein still makes $p_{co,k}$ time-consuming to calculate when $M_k$ goes large. More importantly, the coupling of various parameters, including the number of scheduled sensors $K$, the SIC order index $k$, the number of receive antennas $M_k$, the COP threshold $\beta_{t,k}$, and the jamming probability $\rho$, makes $p_{co,k}$ complicated to analyze. 
In order to circumvent such a difficulty and facilitate the analysis, we focus on a practical requirement of high reliability and low latency. In particular, we examine
the secure transmission in the ultra low COP regime for each sensor. Thereby, we obtain a much more compact expression for $p_{co,k}$ in the following corollary.
\begin{corollary}\label{corollary_cop_low}
In the low COP regime with $p_{co,k}\rightarrow 0$, the COP of $S_k$ can be approximated by
	\begin{equation}\label{cop_low}
	p_{co,k}\approx \frac{\phi\lambda_o\beta_{t,k}^{\delta}}{\pi\lambda_c}\Lambda_k\Xi_{M_k},
	\end{equation}
	where $\Lambda_k = k\binom{K}{k}\sum_{l=0}^{k-1}\binom{k-1}{l}(-1)^{l}\frac{1}{(K-k+l+1)^{2}}$ and $\Xi_{M_k}=1+\sum_{m=1}^{M_k-1}\frac{1}{m!}\prod_{i=0}^{m-1}(i-\delta)$.
\end{corollary}
\begin{proof}
	Please refer to Appendix \ref{proof_corollary_cop_low}.
\end{proof}

It is worth noting that by means of the approximation given above, key parameters such as $K$, $k$, $M_k$, $\beta_{t,k}$, and $\rho$ are decoupled compared to \eqref{cop_int}, and various analytical relationships between the COP and the parameters can be extracted explicitly, some of which are particularly useful for the subsequent optimization of sum secrecy throughput. 
For example, it is clearly shown that $p_{co,k}$ increases as $\beta_{t,k}$ and $\rho$ become larger, as $\lambda_o$ is a monotonically increasing function of $\rho$. Meanwhile, it is as expected that $p_{co,k}$ decreases when $M_k$ grows since $\Xi_{M_k}\in(0,1)$ is monotonically decreasing with $M_k$.

\begin{figure}[!t]
	\centering
	\includegraphics[width= 3.5 in]{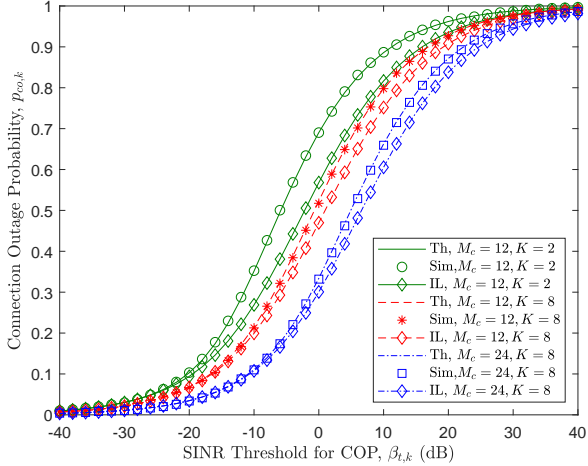}
	\caption{COP $ p_{co,k} $ at $k=1$ v.s. SINR threshold $ \beta_{t,k} $ for different $M_c$ and $K$, with $ P_j = 10 $ dBm, $ \lambda_c = 0.01 $, and $ \rho = 0.05 $. Unless otherwise specified, we always set $ P_a = 10 $ dBm, $ \omega = 0 $ dBm, $ \alpha = 4 $, and $ \lambda_s = 1 $. The labels \{Th, Sim, IL\} refer to the general theoretical result from \eqref{cop_general}, the Monte-Carlo simulated result, and the interference-limited result from \eqref{cop_int}, respectively. }
	\label{figCOP}
\end{figure}

Fig. \ref{figCOP} depicts the COP $ p_{co,k} $ versus the SINR threshold $ \beta_{t,k} $ for different values of the number $M_c$ of FC antennas and the number $K$ of sensor nodes associated with the same FC. Obviously, the Monte-Carlo simulation results are in good agreement with the exact theoretical values.
It is expected that $ p_{co,k} $ monotonically increases with $ \beta_{s,k} $, and the interference-limited results are always smaller than the general ones whereas the gaps are relatively small and even can be negligible as $M_c$, $K$, or $\beta_{t,k}$ goes large enough.  
We can observe that $p_{co,k}$ decreases with a smaller $K$ when $M_c$ is fixed or with a larger $M_c$ for a given $K$.
This indicates that once additional sensor nodes are connected to an FC, the reliability for each sensor node will be degraded, which however can be ameliorated by equipping the FC with more antennas.

\subsection{General Results for SOP}
From a robust secure transmission perspective, we are inclined to consider a worst-case scenario by overestimating the wiretap capability of eavesdroppers. Specifically, we assume that the eavesdroppers have powerful multi-stream decoding capabilities such that they can distinguish multiple data streams received from the scheduled sensors through subtracting interference generated by the superposed signals from each other. In this case, the aggregate interference received at the eavesdroppers only consists of the signals emitted by the jamming sensors.

We further assume that the eavesdroppers employ the optimal linear receiver, i.e., the MMSE receiver, to
improve the quality of the received signals. 
According to the MMSE criterion, the weight vector of the eavesdropper located at $e\in\Phi_e$ for decoding the signal from $S_k$ can be devised in the form of \cite{Gao1998Theoretical}
\begin{equation}\label{weight_mmse}
\bm w_{e,k} = \left(\bm \Psi_{e}+\omega\bm I_{M_e}\right)^{-1}\bm h_{e,s_k},
\end{equation}
where $\bm{\Psi}_{e} \triangleq \sum_{z \in \Phi_{j}} P_{j} \bm{h}_{e,z} \bm{h}_{e,z} ^{H} r_{e,z}^{-\alpha}$.
The SINR of the eavesdropper at $e$ can be given by
\begin{equation}\label{sinr_ek}
{\rm{SINR}}_{e,k}=P_{a} \bm{h}_{e,s_k}^{H} \left(\bm \Psi_{e}+\omega\bm I_{M_e}\right)^{-1} \bm{h}_{e,s_k} r_{e,s_k}^{-\alpha}.
\end{equation}

The SOP $p_{so,k}$ of $S_k$ is defined in \eqref{sop_def} with ${\rm{SINR}}_{e,k}$ given above. The following proposition provides a general result for $p_{so,k}$. 
\begin{proposition}\label{proposition_sop_general}
	The SOP of the secure transmission from the $k$-th nearest sensor $S_k$ to the typical FC is given by
\begin{align}\label{sop_general}
p_{so,k} = 1 - \exp\left(-\pi\lambda_e\sum_{m=1}^{M_e}\sum_{n=0}^{M_e-m}\frac{\zeta_1^{m-1}\zeta_2^n}{(m-1)!n!}\Omega^{\circ}_{u}\right),
\end{align}
where $\Omega^{\circ}_{u}$ has the same form as $\Omega_{\mu}$ defined in Proposition \ref{proposition_cop_general} simply with $\mu = u=\frac{\alpha}{2}(m-1)+n$, $\tau_1 =\zeta_1 ={\omega\beta_{e,k}}/{P_a} $, and $\tau_2=\zeta_2 = \phi\lambda_j\left({P_j\beta_{e,k}}/{P_a}\right)^{\delta}$.
\end{proposition}
\begin{proof}
	Please refer to Appendix \ref{proof_proposition_sop_general}.
\end{proof}

Proposition \ref{proposition_sop_general} clearly shows that the SOP $p_{so,k}$ exponentially increases with the eavesdropper density $\lambda_e$. That is to say, secrecy is severely compromised when facing dense eavesdroppers.
Note that for the spacial case of $\alpha=4$,  $\Omega^{\circ}_{u}$ in \eqref{sop_general} can be recast into the same form of \eqref{omega}, which further leads to a practically closed-form expression for the SOP $p_{so,k}$.

\subsection{Interference-Limited Case for SOP}
Proposition \ref{proposition_sop_general} is not so straightforward for extracting key properties regarding the SOP, which motivates us to seek simplifications. 
To this end, we consider the interference-limited scenario, just as the discussion for the COP, where the receiver noise at eavesdroppers is ignored. Note that this is reasonable since the noise power of eavesdroppers is typically unknown to the sensors.

In the following corollary, we provide an analytically tractable expression for the SOP.
\begin{corollary}\label{corollary_sop}
	For the interference-limited WSN, the SOP of $S_k$ is given by
	\begin{equation}\label{int_sop}
	p_{so,k} = 1 - \exp\left(-\frac{\pi\lambda_e M_e}{\phi\rho\lambda_i}\left(\frac{P_a}{P_j\beta_{e,k}}\right)^{\delta}\right).
	\end{equation}
\end{corollary}
\begin{proof}
	The result follows easily by plugging $\omega=0$ into \eqref{sop_general} and leveraging some algebraic operations.
\end{proof}

Corollary \ref{corollary_sop} reveals that the SOP $p_{so,k}$ exponentially increases with $M_e$ but decreases with $\rho$ and $\beta_{e,k}$. This indicates that secrecy performance is dramatically degraded if eavesdroppers use a large number $M_e$ of receiving antennas, whereas it can be significantly ameliorated by making a larger fraction $\rho$ of sensors send jamming signals and choosing a larger rate redundancy $R_{e,k}$ for channel coding. 
Note that this is fundamentally different from the case of COP where increasing the jammer fraction $\rho$ becomes harmful, which reflects an intrinsic trade-off between reliability and secrecy when introducing jamming signals.

\begin{figure}[!t]
	\centering
	\includegraphics[width= 3.5 in]{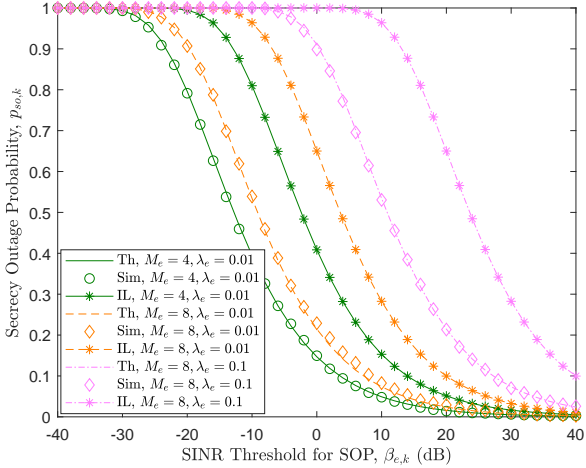}
	\caption{SOP $ p_{so,k} $ at $k=1$ v.s. SINR threshold $ \beta_{e,k} $ for different $M_e$ and $\lambda_e$, with $ P_j = 10 $ dBm, $ \lambda_c = 0.01 $, $ K = 3 $, and $ \rho = 0.05 $. The labels \{Th, Sim, IL\} refer to the general theoretical result from \eqref{sop_general}, the Monte-Carlo simulated result, and the interference-limited result from \eqref{int_sop}, respectively. }
	\label{figSOP}
\end{figure}

The monotonicity of the SOP $ p_{so,k} $ w.r.t. the  SINR threshold $ \beta_{e,k} $, the number $M_e$ of eavesdropping antennas, and the density $\lambda_e$ of eavesdroppers is validated by both numerical and simulated results as shown in Fig. \ref{figSOP}. 
Different from the situation of COP, the interference-limited SOP is apparently larger than the general one. This implies that the interference-limited SOP embodies an overestimation of eavesdropping capability, which is generally preferred when investigating physical layer security for the purpose of robustness designs.

\section{Maximization of Sum Secrecy Throughput}
This section maximizes the sum secrecy throughput for a typical FC in the large-scale WSN with random multiple access, by jointly determining the optimal parameters, including the codeword rate $R_{t,k}$ and the secrecy rate $R_{s,k}$ of the wiretap code for the $K$ scheduled sensors, and the jamming probability $\rho$ of the proposed random jamming scheme.
Recalling the definition of sum secrecy throughput in \eqref{sst_def}, the optimization problem can be formulated as 
\begin{subequations}\label{sst_max}
	\begin{align}
	\max_{R_{t,k},R_{s,k},\forall k\in\mathcal{K},\rho} ~&\mathcal{T} = \sum_{k=1}^K R_{s,k}\left(1-p_{co,k}\right)\\
	{\rm s.t.}~~ \label{sst_max_c1}
	&p_{co,k}\le\sigma,~ \forall k\in\mathcal{K},\\
	\label{sst_max_c2}
	&p_{so,k}\le\epsilon,~\forall k\in\mathcal{K},\\
	\label{sst_max_c3}
	&0\le R_{e,k} = R_{t,k}-R_{s,k},~ \forall k\in\mathcal{K},\\
	\label{sst_max_c4}
	&0\le \rho \le 1.
	\end{align}
\end{subequations}
Note that constraints \eqref{sst_max_c1} and \eqref{sst_max_c2} describe the reliability and secrecy requirements, respectively;
constraints \eqref{sst_max_c3} and  \eqref{sst_max_c4} are imposed by the wiretap code scheme and the random jamming scheme, respectively. 

The original problem \eqref{sst_max} can be decomposed into the following two subproblems. 
\begin{itemize}
	\item[1)] We first design the optimal $R_{s,k}$ and $R_{e,k}$ (or $R_{t,k}$) for the $k$-th scheduled sensor to maximize its secrecy throughput $\mathcal{T}_k \triangleq R_{s,k}\left(1-p_{co,k}\right)$ conditioned on a fixed $\rho$, as
		\begin{align}\label{opt_sst_max1}
		\max_{R_{s,k},R_{e,k}} &\mathcal{T}_k, ~\forall k\in\mathcal{K},
		~~{\rm s.t.}~~\eqref{sst_max_c1}-\eqref{sst_max_c3}.
		\end{align}
	\item[2)] With the resultant maximal $\mathcal{T}_k$ for $k\in\mathcal{K}$, we then design the optimal $\rho$ to maximize $\mathcal{T}$ expressed in \eqref{sst_max}, i.e.,
	\begin{align}\label{opt_sst_max2}
		\max_{\rho} ~\mathcal{T},~~
		{\rm s.t.}~~
			\eqref{sst_max_c4}.
	\end{align}
\end{itemize} 

In the following two subsections, we first discuss an optimal design scheme in which the optimal $R_{s,k}$ and $R_{e,k}$ can be efficiently calculated by the bisection method while the optimal $\rho$ can only be obtained by one-dimensional search. We then examine a sub-optimal scheme for the purpose of a low computational complexity, where closed-form expressions are derived for the optimal $R_{s,k}$ and $R_{e,k}$, and $\mathcal{T}$ is proved to be quasi-concave w.r.t. $\rho$ such that the optimal $\rho$ can be computed using the bisection method. 

\subsection{Optimal Design}
Based on the above discussion, we first examine the subproblem \eqref{opt_sst_max1} and design the optimal $R_{e,k}$ and $R_{s,k}$ successively. Consider a fixed $R_{s,k}$, it is apparent from \eqref{cop_def} that the COP $p_{co,k}$ monotonically decreases with $R_{e,k}$ as $R_{t,k}=R_{s,k}+R_{e,k}$. 
This suggests that the optimal $R_{e,k}$ for maximizing $\mathcal{T}_{k}$ should be the minimal $R_{e,k}$ while satisfying the secrecy constraint $p_{so,k}\le\epsilon$. Note that  $p_{so,k}$ decreases with $R_{e,k}$ (since $\beta_{e,k}=2^{R_{e,k}}-1$ shown in \eqref{sop_def}), the optimal $R_{e,k}$ is given as the inverse of
$p_{so,k}(R_{e,k})$ at $\epsilon$, which is
\begin{equation}
R^*_{e,k} = p_{so,k}^{-1}(\epsilon).
\end{equation}
Obviously, $R^*_{e,k}$ monotonically decreases with $\epsilon$, which means that a larger rate redundancy is required to combat the eavesdropper in order to meet a more rigorous secrecy constraint. 
Although it is intractable to
express $R^*_{e,k}$ in an explicit form due to the complicated expression of $p_{so,k}$, the value of $R^*_{e,k}$ can be efficiently obtained through bisection search with the equation $p_{so,k}(R_{e,k})=\epsilon$.

For designing the optimal $R_{s,k}$, we focus on the low COP regime and substitute the approximate COP $p_{co,k}$ given in \eqref{cop_low} into problem \eqref{opt_sst_max1}. Moreover, since  $R_{t,k} = R_{e,k}^{*}+R_{s,k}\Rightarrow\beta_{t,k}=\beta_{e,k}^{*}+(1+\beta_{e,k}^{*})\beta_{s,k}$ with $\beta_{s,k}\triangleq 2^{R_{s,k}}-1$, problem \eqref{opt_sst_max1} can be equivalently translated into  
\begin{align}\label{st_max}
\max_{0\le\beta_{s,k}\le \beta^{\rm max}_{s,k}}~ \mathcal{T}_k=\left(1-A_k\left(\beta_{s,k}+B_k\right)^{\delta}\right)\log_2(1+\beta_{s,k}),
\end{align}
where $A_k\triangleq\frac{\phi\lambda_o}{\pi\lambda_c}\Lambda_k\Xi_{M_k}$ with $\Lambda_k$ and $\Xi_{M_k}$ defined in Corollary \ref{corollary_cop_low}, $B_k \triangleq \frac{\beta^*_{e,k}}{1+\beta^*_{e,k}}$, and $\beta^{\rm max}_{s,k}\triangleq A_k^{-\alpha/2}-B_k$.
It is noteworthy that $\beta^{\rm max}_{s,k}$ is introduced to guarantee a non-negative value of $\mathcal{T}_k$.
The solution to the above problem is provided by the following proposition.
\begin{proposition}\label{proposition_st_opt}
	The secrecy throughput $\mathcal{T}_k$ given in \eqref{st_max} is a concave function of $\beta_{s,k}$, and the optimal $\beta_{s,k}^*$ maximizing $\mathcal{T}_k$ satisfies the following equation, 
	\begin{equation}\label{opt_rs}
	\frac{d\mathcal{T}_k}{d\beta_{s,k}}|_{\beta_{s,k}=\beta_{s,k}^*} = 0,
	\end{equation}
	i.e., it is the unique zero-crossing point $\beta_{s,k}$ of the derivative $\frac{d\mathcal{T}_k}{d\beta_{s,k}}$ given below
	\begin{equation}\label{opt_bs_db_eqn}
\frac{d\mathcal{T}_k}{d\beta_{s,k}}= \frac{1-A_k(\beta_{s,k}+B_k)^{\delta}}{(1+\beta_{s,k})\ln2}-\frac{A_k\delta\log_2(1+\beta_{s,k})}{(\beta_{s,k}+B_k)^{1-\delta}}.
	\end{equation}
\end{proposition}
\begin{proof}
	For brevity, the subscripts of $\beta_{s,k}$, $A_k$, and $B_k$ are omitted.
	It is intractable to prove the concavity of $\mathcal{T}_k$ on $\beta$ by determining the sign of the second-order derivative $\frac{d^2\mathcal{T}_k}{d\beta^2}$.
Instead, it can be easily confirmed that the two boundary values of $\beta$ yield $\frac{d\mathcal{T}_k}{d\beta}|_{\beta=0}>0$ and $\frac{d\mathcal{T}_k}{d\beta}|_{\beta=\beta^{\rm max}_{s,k}}<0$.
Combined with the fact that $\mathcal{T}_k$ is continuously differentiable on $\beta$, there at least exists one zero-crossing point of  $\frac{d\mathcal{T}_k}{d\beta}$.
Let $\beta^{\circ}$ denote an arbitrary one such that $\frac{d\mathcal{T}_k}{d\beta}|_{\beta=\beta^{\circ}}=0$, and then the second-order derivative $\frac{d^2\mathcal{T}_k}{d\beta^2}$ at $\beta=\beta^{\circ}$ can be calculated as 
\begin{align}
&\frac{d^2\mathcal{T}_k}{d\beta^2}|_{\beta=\beta^{\circ}}=-\frac{2A\delta(\beta^{\circ}+B)^{\delta-1}}{(1+\beta^{\circ})\ln2}-\frac{1-A(\beta^{\circ}+B)^{\delta}}{(1+\beta^{\circ})^2\ln2}+\nonumber\\
&\quad\quad\quad\quad\quad\quad\quad\quad\quad\quad A\delta(1-\delta)(\beta^{\circ}+B)^{\delta-2}\log_2(1+\beta^{\circ})\nonumber\\
&\stackrel{\mathrm{(a)}}=\frac{A\delta(\beta^{\circ}+B)^{\delta-2}}{\ln2}\bigg((1-\delta)\ln(1+\beta^{\circ})-\nonumber\\
&\quad\quad\quad\quad\quad\quad\quad\quad\quad\quad\quad\quad\quad
\frac{(\beta^{\circ}+B)\left[2+\ln(1+\beta^{\circ})\right]}{1+\beta^{\circ}}\bigg)\nonumber\\
&\stackrel{\mathrm{(b)}}<
\frac{A\delta(\beta^{\circ}+B)^{\delta-2}}{\ln2}\left(\ln(1+\beta^{\circ})-\frac{\beta^{\circ}\left[1+\ln(1+\beta^{\circ})\right]}{1+\beta^{\circ}}\right)\nonumber\\
&=\frac{A\delta(\beta^{\circ}+B)^{\delta-2}}{(1+\beta^{\circ})\ln2}\left[{\ln(1+\beta^{\circ})}-{\beta^{\circ}}\right]\nonumber\\
&\stackrel{\mathrm{(c)}}\leq0,
\end{align}
where $\mathrm{(a)}$ holds by noting that $\frac{d\mathcal{T}_k}{d\beta}|_{\beta=\beta^{\circ}}=0\Rightarrow \frac{1-A(\beta^{\circ}+B)^{\delta}}{(1+\beta^{\circ})^2}=A\delta(\beta^{\circ}+B)^{\delta-1}\ln(1+\beta^{\circ})$ in \eqref{opt_bs_db_eqn}, $\mathrm{(b)}$ gives an upper bound, and $\mathrm{(c)}$ follows from the fact $\ln(1+\beta^{\circ})\le \beta^{\circ}$.
The above result indicates that $\mathcal{T}_k$ is a quasi-concave function of $\beta$ \cite{Boyd2004Convex}, and $\beta^{\circ}$ is the unique zero-crossing point of $\frac{d\mathcal{T}_k}{d\beta}$ and is also the solution to problem \eqref{st_max}.
\end{proof}

Due to the quasi-concavity of $\mathcal{T}_k$ on $\beta_{s,k}$, the value of the optimal $\beta_{s,k}^*$ can be efficiently calculated using the bisection method with \eqref{opt_rs}.
After that, by substituting the obtained optimal $R_{e,k}^*$ and $R_{s,k}^*$ for $k\in\mathcal{K}$ into \eqref{sst_max}, the optimal jamming probability $\rho^*$ can be numerically searched by solving problem \eqref{opt_sst_max2}.

\begin{figure}[!t]
	\centering
	\includegraphics[width= 3.5 in]{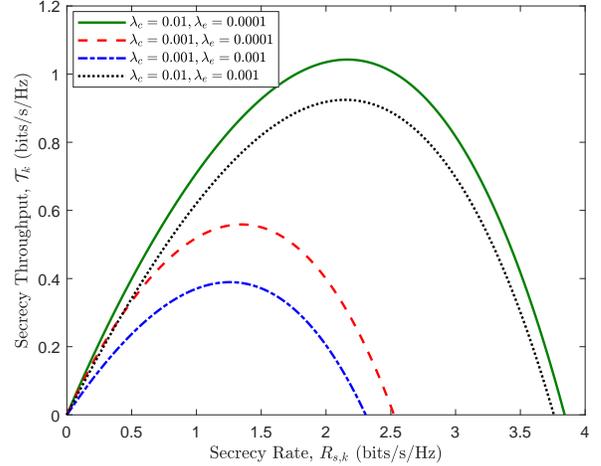}
	\caption{Secrecy throughput $ \mathcal{T}_k $ at $ k = 1 $ v.s. secrecy rate $ R_{s,k} $ for different $\lambda_c$ and $\lambda_e$, with $ P_j = 0 $ dBm, $ \rho = 0.01 $, $ M_c = 16 $, $ M_e  = 2$, $ K = 4 $, $ \epsilon = 0.1 $, and $ \sigma = 0.1 $.}
	\label{figTk}
\end{figure}

Fig. \ref{figTk} depicts the secrecy throughput $ \mathcal{T}_k $ as a function of the secrecy rate $ R_{s,k} $. Just as analyzed in Proposition \ref{proposition_st_opt}, we see that $ \mathcal{T}_k $ indeed initially increases and then decreases with $ R_{s,k} $, and there is a unique $ R_{s,k} $ for maximizing $ \mathcal{T}_k $. It is expected that increasing the density $\lambda_e$ of eavesdroppers is harmful to the improvement of secrecy throughput. We also observe that, with our proposed random access scheme, the secrecy throughput can be dramatically increased with the increase of FC density $ \lambda_c $. This seemingly counter-intuitive result can be understood if one can realize that although deploying more FCs will accommodate more sensor nodes resulting in more severe network interference, it will reduce the distances between the FC and its associated sensor nodes making transmission reliability significantly enhanced.

\subsection{Sub-optimal Design}
Note that the optimal code rates and optimal jamming probability for the optimal design scheme can only be obtained numerically via bisection search or exhaustive search, which not only results in high computational complexity but also makes it difficult to develop useful insights into practical system designs. 
To this end, this subsection examines a sub-optimal solution to problem \eqref{sst_max} by focusing on the sum secrecy throughput with COP constraints $p_{co,k}=\sigma$ and SOP constraints $p_{so,k}=\epsilon$ for $k\in\mathcal{K}$. 
The corresponding sum secrecy throughput can be written as 
\begin{equation}\label{sub_st}
\mathcal{T}=(1-\sigma)\sum_{k=1}^K\log_2\frac{1+\beta_{t,k}^*}{1+\beta_{e,k}^*},
\end{equation}
where $\beta_{t,k}^*$ and $\beta_{e,k}^*$ are the unique roots of $\beta_{t,k}$ and $\beta_{e,k}$ of the equations $p_{co,k}(\beta_{t,k})=\sigma$ and $p_{so,k}(\beta_{e,k})=\epsilon$, respectively.
The rationality of devising the sub-optimal design is that the values of $\sigma$ and $\epsilon$ are generally set small enough to  guarantee a high level of reliability and secrecy. 
Moreover, when the COP and SOP constraints can be controlled, we are able to maximize $\mathcal{T}$
by finding the optimal values of $\sigma$ and $\epsilon$ in \eqref{sub_st}.

By solving the equations $p_{co,k}(\beta_{t,k}^*)=\sigma$ and $p_{so,k}(\beta_{e,k}^*)=\epsilon$ recalling $p_{co,k}$ in \eqref{cop_low} and $p_{so,k}$ in \eqref{int_sop}, we can provide closed-form expressions for $\beta_{t,k}^*$ and $\beta_{e,k}^*$ (also $R_{t,k}^*$ and $R_{e,k}^*$) by the following proposition.
\begin{proposition}\label{value_bt_be}
	The values of $\beta_{t,k}^*$ and $\beta_{e,k}^*$ for $k\in\mathcal{K}$ that satisfy $p_{co,k}(\beta_{t,k}^*)=\sigma$ and $p_{so,k}(\beta_{e,k}^*)=\epsilon$ can be respectively given by
	\begin{equation}\label{bt}
	\beta_{t,k}^*=\left(\frac{\sigma\pi\lambda_c}{\lambda_a\phi\Lambda_k\Xi_{M_k}}\right)^{\frac{\alpha}{2}}\left[1+\frac{\lambda_iP_j^{\delta}}{\lambda_aP_a^{\delta}}\rho\right]^{-\frac{\alpha}{2}},
	\end{equation}
	\begin{equation}\label{be}
	\beta_{e,k}^*=\frac{P_a \rho^{-\frac{\alpha}{2}}}{P_j}\left(\frac{\pi\lambda_eM_e }{\phi\lambda_i \ln\frac{1}{1-\epsilon} }\right)^{\frac{\alpha}{2}}.
	\end{equation}
\end{proposition}

The next step is to design the optimal jamming probability $\rho$ to maximize the sum secrecy throughput $\mathcal{T}$.
Before proceeding to the optimization procedure, we introduce three auxiliary variables $X_k\triangleq \left(\frac{\sigma\pi\lambda_c}{\lambda_a\phi\Lambda_k\Xi_{M_k}}\right)^{\frac{\alpha}{2}}$, $Y\triangleq \frac{\lambda_iP_j^{\delta}}{\lambda_aP_a^{\delta}}$, and $Z\triangleq \frac{P_a}{P_j}\left(\frac{\pi\lambda_eM_e }{\phi\lambda_i \ln\frac{1}{1-\epsilon} }\right)^{\frac{\alpha}{2}}$ for \eqref{bt} and \eqref{be}, with which $\mathcal{T}$ in \eqref{sub_st} can be recast as a function of a single variable $\rho$ as given below:
\begin{equation}\label{sub_st_xyz}
\mathcal{T}=\sum_{k=1}^K \mathcal{T}_k=(1-\sigma)\sum_{k=1}^K\log_2\frac{1+X_k(1+Y\rho)^{-\frac{\alpha}{2}}}{1+Z\rho^{-\frac{\alpha}{2}}}.
\end{equation}

Remarkably, $X_k$, $Y$, and $Z$ have clear physical significance. 
Specifically, $X_k$ can be interpreted as the ability of boosting the achievable rate for the legitimate channel. 
For example, a looser COP constraint (a larger $\sigma$) and a larger number of receive antennas at the FC side (a smaller $\Xi_{M_k}$) will increase $X_k$ and are beneficial for improving transmission reliability. 
Similarly, $Z$ can be translated as the wiretapping capability and $Y$ reflects the jamming power level a sensor can afford. 

It is obvious that in order to guarantee a non-negative secrecy throughput for each sensor, i.e., $\mathcal{T}_k\ge 0$ for $k\in\mathcal{K}$, $X_k(1+Y\rho)^{-\frac{\alpha}{2}}>Z\rho^{-\frac{\alpha}{2}}$ must be satisfied from \eqref{sub_st_xyz}, which produces $\rho\ge \left(\left(\frac{X_k}{Z}\right)^\delta-Y\right)^{-1}$.
In other words, we should ensure $\rho\ge \rho_{\rm min}\triangleq\max_{k\in\mathcal{K}}\left(\left(\frac{X_k}{Z}\right)^\delta-Y\right)^{-1}$. Hence, the optimal $\rho^*$ maximizing $\mathcal{T}$ can be obtained by solving the following equivalent problem.
	\begin{align}\label{sub_st_max}
	\max_{\rho_{\rm min}\le\rho\le 1} ~T(\rho)\triangleq \sum_{k=1}^K\ln\frac{1+X_k(1+Y\rho)^{-\frac{\alpha}{2}}}{1+Z\rho^{-\frac{\alpha}{2}}}.
	\end{align} 
	
Although the above problem is not convex, in the following proposition we introduce a derivative reconstruction method to prove that $T(\rho)$ is actually first-increasing-then-decreasing w.r.t. $\rho$ such that the optimal $\rho^*$ maximizing $T(\rho)$ must be unique.
\begin{proposition}\label{proposition_subopt_rho}
	The objective function $T(\rho)$ in \eqref{sub_st_max} initially increases and then decreases with an increasing $\rho$, and the optimal $\rho^*$ that maximizes $T(\rho)$ is provided as
	\begin{align}\label{opt_rho}
	\rho^*=
	\begin{cases}
	\rho_{\rm min}, & Z<\frac{Y\rho_{\rm min}^{\alpha/2+1}\left(1-\kappa(\rho_{\rm min})\right)}{1+ Y\rho_{\rm min}\kappa(\rho_{\rm min})}\\
	1,& Z\ge \frac{Y(1-\kappa(1))}{1+Y\kappa(1)}\\
	\rho^{\circ}, & {\rm otherwise} 
	\end{cases}
	\end{align}
	where $\kappa(\rho)\triangleq\frac{1}{K}\sum_{k=1}^K\frac{1}{X_k\left(1+Y\rho\right)^{-\alpha/2}+1}<1$, and $\rho^{\circ}$ is the unique root $\rho$ of the equation $G(\rho)=0$ with $G(\rho)$ being a monotonically decreasing function of $\rho$ given by
	\begin{equation}\label{G}
	G(\rho)=1 + Y\rho \kappa(\rho)-\frac{Y\rho^{\alpha/2+1}}{Z}\left(1-\kappa(\rho)\right).
	\end{equation}
\end{proposition}
\begin{proof}
	Please refer to Appendix \ref{proof_proposition_subopt_rho}.
\end{proof}

Some observations regarding the design of the optimal jamming probability $\rho^*$ can be obtained from Proposition \ref{proposition_subopt_rho}:

1) As previously explained, the variable $Z$ actually embodies the advantage that eavesdroppers can perform attacks. 
When such advantage is marginal, i.e., $Z<\frac{Y\rho_{\rm min}^{\alpha/2+1}\left(1-\kappa(\rho_{\rm min})\right)}{1+ Y\rho_{\rm min}\kappa(\rho_{\rm min})}$, it is not necessary to activate too many sensors to radiate jamming signals to confuse the eavesdroppers. 
In this case, we can simply set the optimal $\rho^*$ to its minimal achievable value $\rho_{\rm min}$. 
Nevertheless, the maximization of sum secrecy throughput comes at the expense of unfairness, since there exists at least one sensor whose secrecy throughput would be reduced to zero, i.e., the sensor with  index $k^{\circ}=\arg\max_{k\in\mathcal{K}}\left(\left(\frac{X_k}{Z}\right)^\delta-Y\right)^{-1}$.

2) If the eavesdroppers' superiority exceeds a certain level, i.e., $Z\ge \frac{Y(1-\kappa(1))}{1+Y\kappa(1)}$, all the idle sensors have to be mobilized for anti-eavesdropping.
Hence, $\rho^*=1$ is optimal for the sum secrecy throughput maximization.

 3) Beyond the above two situations, we should properly set the jamming probability $\rho$ to strike a good balance between throughput and secrecy. Although an explicit form of $\rho^{\circ}$ cannot be derived, we can still develop some useful properties on $\rho^{\circ}$ to guideline practical designs, as summarized in the following corollary.
 \begin{corollary}\label{corollary_opt_rho}
 	The optimal jamming probability $\rho^{\circ}$ decreases with the maximal endurable COP $\sigma$ and SOP $\epsilon$, the number $M_c$ of antennas at the FC side, the sensor density $\lambda_s$, and the ratio $P_j/P_a$ of jamming power to transmit power of a sensor, while increases with the FC density $\lambda_c$, the number $K$ of sensors associated with each FC, the eavesdropper density $\lambda_e$, and the number $M_e$ of antennas at the eavesdropper side.
 \end{corollary}
 \begin{proof}
 	Please refer to Appendix \ref{proof_corollary_opt_rho}.
 \end{proof}

\begin{figure}[!t]
	\centering
	\includegraphics[width= 3.5 in]{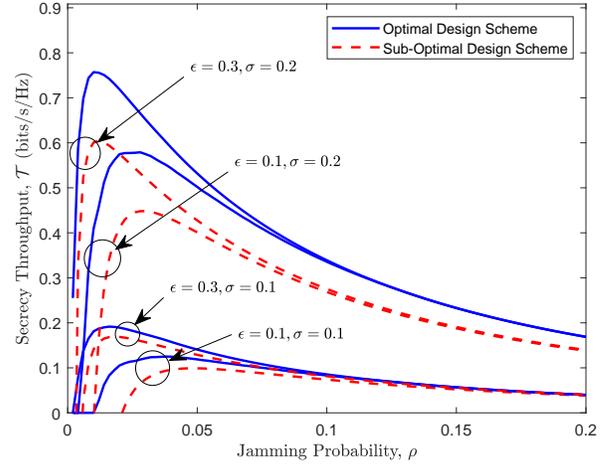}
	\caption{Secrecy throughput $ \mathcal{T} $ v.s. jamming probability $ \rho $ for different $\epsilon$ and $\sigma$, with $ P_j = 1 $ dBm, $ \lambda_c = 0.01 $, $ M_c = 16 $, $ K = 4 $, $ M_e = 2$, and $ \lambda_e = 0.0001 $.}
	\label{figTrho}
\end{figure}

Fig. \ref{figTrho} illustrates the secrecy throughput $\mathcal{T}$ as a function of jamming probability $\rho$ for both the optimal and sub-optimal schemes. As proved previously, $ \mathcal{T} $ first increases and then decreases as $ \rho $ increases, and there exists a unique optimal $ \rho $ that maximizes $ \mathcal{T} $. 
We show that as either $\epsilon$ or $\sigma$ becomes larger, the optimal $\rho$ becomes smaller producing a higher $ \mathcal{T} $ for both the optimal and sub-optimal schemes, which validates Corollary \ref{corollary_opt_rho}.  
The reason behind is that facing a looser SOP constraint (a larger $\epsilon$), fewer sensor nodes are required to send jamming signals against eavesdropping; meanwhile, when a larger COP $\sigma$ can be tolerable, activating less jammers significantly benefits secrecy throughput via supporting a much larger secrecy rate. 
We find that as $\rho$ increases, the two curves with different $\epsilon$'s but identical $\sigma$ merge. This implies that the jamming probability is sufficiently large to defeat eavesdroppers such that 
the secrecy throughput performance is less sensitive to the variation of the SOP constraints. 
It is interesting to observe that the gap between optimal and sub-optimal schemes decreases obviously as the COP threshold $\sigma$ reduces. 
This is because for a more stringent COP constraint, the feasible region of the secrecy rate for the optimal scheme shrinks such that the optimal secrecy rate  maximizing secrecy throughput approaches that of the sub-optimal scheme.

\begin{figure}[!t]
	\centering
	\includegraphics[width= 3.5 in]{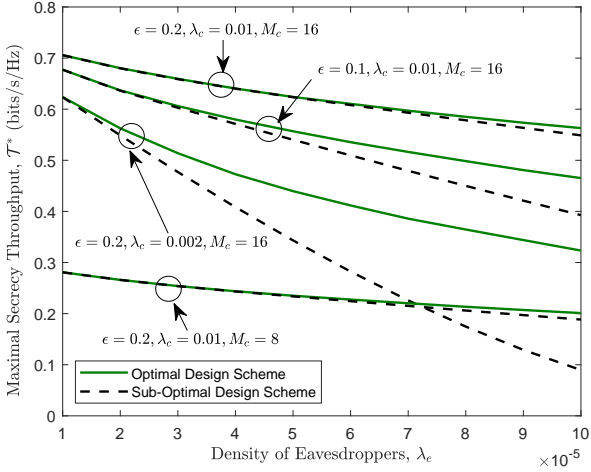}
	\caption{Maximal secrecy throughput $ \mathcal{T}^* $ v.s. eavesdropper density $\lambda_e$ for different $\epsilon$, $\lambda_c$, and $M_c$, with $ P_j = 1 $ dBm, $ K = 4 $, $ M_e = 2 $, and $ \sigma = 0.2 $.}
	\label{figTstar}
\end{figure}

Fig. \ref{figTstar} plots the maximal secrecy throughput $ \mathcal{T}^* $ of both optimal and sub-optimal schemes. 
It is easy to understand that $\mathcal{T}^*$ decreases with increasing density $\lambda_e$ of eavesdroppers and grows with increasing SOP threshold $\epsilon$, density $\lambda_c$ of FCs, and number $M_c$ of antennas at the FC side.
We also show that the gap between the optimal and sub-optimal schemes decreases as 
$ \epsilon $ or $ \lambda_c$ increases or as $\lambda_e$ decreases.
The underlying reason is that for these situations, adopting a larger secrecy rate can be more beneficial for maximizing secrecy throughput even sacrificing the reliability. This would make the resultant COP for the optimal scheme approach the COP threshold $\sigma$ which is the COP for the sub-optimal case, and hence the secrecy throughput performance for the two schemes becomes similar.

\section{Conclusions}
Physical layer security was investigated for a large-scale WSN with random multiple access under a stochastic geometry framework. 
An uncoordinated jamming scheme was devised to thwart the randomly distributed eavesdroppers.
Analytical expressions were derived for both the COP and SOP of the secure data delivery from sensors to a typical FC against eavesdropping.
Afterwards, the optimal wiretap code rates and the jamming probability were jointly designed to maximize the sum secrecy throughput subject to both COP and SOP constraints, with both optimal and sub-optimal algorithms examined.
Furthermore, some insights into how the optimal parameters should be adjusted to the communication environment and performance requirements were provided. 
Numerical results were presented to validate the theoretical fundings.
In particular, it was shown that for a stringent COP constraint or a loose SOP constraint, the performance gap between optimal and sub-optimal schemes becomes insignificant meaning that the sub-optimal scheme can be adopted as a low-complexity alternative to the optimal one.

\appendix
\subsection{Proof of Proposition \ref{proposition_cop_general}}\label{proof_proposition_cop_general}
Let $s\triangleq\frac{\beta_{t,k}L_k^{\alpha}}{P_a}$ and $I = I_a+I_j$, the COP $p_{co,k}$ can be computed by substituting \eqref{sinr_ok} into \eqref{cop_def},
\begin{align}\label{pco_appendix}
p_{co,k}  &=1- \mathbb{E}_{L_k}\mathbb{E}_{I} \left[ \mathbb{P}\left\{\|\bm U_k^{\rm T}\bm{h}_{o,s_k}^{\dagger}\|^2\geq {s}(I+\omega)\right\}\right]\nonumber\\
&\stackrel{\mathrm{(a)}}
= 1-\mathbb{E}_{L_k} \mathbb{E}_{I}\left[e^{-s(I+\omega)}
\sum_{m=0}^{M_k-1}\frac{s^m(I+\omega)^m}
{m!}\right]\nonumber\\
&=1- \mathbb{E}_{L_k} \left[e^{-s \omega}\sum_{m=0}^{M_k-1}\sum_{p=0}^{m}\binom{m}{p}\frac{\omega^{m-p}s^m}
{m!}\mathbb{E}_{I}\left[I^pe^{-s I}\right]\right]\nonumber\\
&\stackrel{\mathrm{(b)}}=1- \mathbb{E}_{L_k}\left[e^{-s \omega}
\sum_{m=0}^{M_k-1}\sum_{p=0}^{m}\binom{m}{p}\frac{\omega^{m-p}s^m}
{(-1)^pm!}\frac{d^p}{ds^p}\mathcal{L}_{I}(s)\right],
\end{align}
where (a) is due to $\|\bm U_k^{\rm T}\bm{h}_{o,s_k}^{\dagger}\|^2\sim {\rm Gamma}(M_k,1)$, and (b) follows from the Laplace transform property $t^nf(t)\overset{\mathcal{L}}\leftrightarrow(-1)^{n}\frac{ d^{n} }{d s^{n}}\mathcal{L}_{f(t)}(s)$.
Due to the independence of $I_a$ and $I_j$, the Laplace transform $\mathcal{L}_{I}(s)$ can be expressed as \cite[eqn. (8)]{Haenggi2009Stochastic}
\begin{align}\label{laplace}
\mathcal{L}_{I}(s)&=\mathbb{E}_{I_a+I_j}\left[e^{-s(I_a+I_j) }\right]=\mathcal{L}_{I_{a}}(s) \mathcal{L}_{I_{j}}(s)\nonumber\\
& =e^{-\phi \left(\lambda_aP_a^\delta+\lambda_jP_j^\delta\right) s^\delta}=e^{-\phi \lambda_o (P_as)^\delta},
\end{align}
where $\lambda_o = \lambda_a+\left({P_j}/{P_a}\right)^\delta \lambda_j$.
The $p$-order derivative $\frac{d^p}{ds^p}\mathcal{L}_{I}(s)$ can be obtained by \cite[Eq. (51)]{{Hunter08Transmission}}
\begin{equation}\label{laplace_p}
\frac{d^p}{ds^p}\mathcal{L}_{I}(s) = \frac{e^{-\phi  \lambda_o (P_as)^\delta}}{(-s)^{p}} \sum_{n=1}^{p}\left[\delta \phi \lambda_o (P_as)^\delta \right]^{n} \Upsilon_{p,n}.
\end{equation}
Substituting \eqref{laplace_p} into \eqref{pco_appendix} with $s={\beta_{t,k}L_k^{\alpha}}/{P_a}$ yields 
\begin{align}\label{pcok}
p_{co,k} = &1-\sum_{m=0}^{M_k-1}\sum_{p=0}^{m}\binom{m}{p}\left(\frac{\omega\beta_{t,k}}{P_a}
\right)^{m-p}\sum_{n=1}^{p}\frac{\left(\delta \phi \lambda_o \beta_{t,k}^\delta \right)^{n}}{m!} 
\nonumber\\
& \Upsilon_{p,n}\underbrace{\mathbb{E}_{L_k}\left[L_k^{\alpha(m-p)+2n}e^{- \frac{\omega\beta_{t,k}}{P_a}L_k^{\alpha}-\phi  \lambda_o\beta_{t,k}^\delta L_k^2}
	\right]}_{\mathcal{I}_k}.
\end{align}
The term $\mathcal{I}_k$ in \eqref{pcok} can be calculated as
\begin{align}
\mathcal{I}_k &= \int_0^{\infty} r^{\alpha(m-p)+2n}e^{- \frac{\omega\beta_{t,k}}{P_a}r^{\alpha}-\phi  \lambda_o\beta_{t,k}^\delta r^2} 	f_{L_{k}}(r)dr\nonumber\\
 &\stackrel{\mathrm{(a)}}=\pi\lambda_ck\binom{K}{k}\sum_{l=0}^{k-1}\binom{k-1}{l}(-1)^l \underbrace{\int_0^{\infty}x^{\mu}e^{-\tau_1x^{\alpha/2}-\tau_2 x}dx}_{\Omega_{\mu}},
\end{align}
where (a) follows from invoking the PDF $f_{L_{k}}(r)$ of $L_k$ given in \eqref{lk_pdf} along with the substitution $r^2\rightarrow x$. The proof can be completed after discussing the cases $p=0$ and $p\neq 0$.

\subsection{Proof of Corollary \ref{corollary_cop_int}}\label{proof_corollary_cop_int}
Plugging $\omega=0$ into \eqref{cop_general} yields
\begin{align}\label{cop_app}
	p_{co,k}&=  1 -\pi\lambda_ck\binom{K}{k}\sum_{l=0}^{k-1}\binom{k-1}{l}\sum_{m=0}^{M_k-1}	\frac{(-1)^{l}}{m!}\times	\nonumber\\
& 
\left[\bm 1_{m=0}\Omega_{0}+\bm 1_{m\neq0}\sum_{n=1}^m\left(\delta\phi\lambda_o\beta_{t,k}^{\delta}\right)^n\Omega_{n}\Upsilon_{m,n}\right].
\end{align}
Recalling $\Omega_{\mu}$ defined in Proposition \ref{proposition_cop_general}, it is easy to obtain that 
$\Omega_0 = {1}/{\tau_2}$ and $\Omega_n = {n!}/{\tau_2^{n+1}}$.
Then, the proof is completed.

\subsection{Proof of Corollary \ref{corollary_cop_low}}\label{proof_corollary_cop_low}
To begin with, let us revisit $p_{co,k}$ in \eqref{pco_appendix} and plug \eqref{laplace_p} with $s = {\beta_{t,k}L_k^{\alpha}}/{P_a}$ and $\omega = 0$ into \eqref{pco_appendix}.
Then we obtain
\begin{align}\label{pco_appendix_2}
&p_{co,k} = 1- \mathbb{E}_{L_k}\left[
\sum_{m=0}^{M_k-1}\frac{(-s)^m}
{m!}\frac{d^m}{ds^m}\mathcal{L}_{I}(s)\right] =1 - \nonumber\\
& \mathbb{E}_{L_k}\left[\underbrace
{e^{-\psi_k\beta_{t,k}^{\delta}}\left(1+\sum_{m=1}^{M_k-1}\sum_{n=1}^m\frac{\left(\delta\psi_k\beta_{t,k}^{\delta}\right)^n}{m!}\Upsilon_{m,n}\right)}_{\mathcal{Q}_k}\right],
\end{align}
where $\psi_k = \phi\lambda_o L_k^2$.
We observe that as $\psi_k\rightarrow 0$, the term $\mathcal{Q}_k\rightarrow 1$ for any $L_k$, which finally leads to $p_{co,k}\rightarrow 0$. Note that the asymptotic region $\psi_k\rightarrow 0$ reflects all possible situations where parameters including but not limited to $\lambda_c$, $\lambda_j$, $P_j$, and $L_k$ may produce a sufficiently small COP $p_{co,k}$.
Invoking the first-order Taylor
expansion with $e^{-\psi_k\beta_{t,k}^{\delta}}$ in \eqref{pco_appendix_2} around 
$\psi_k= 0$ and discarding the high order terms $\mathcal{O}\left(\psi_k^2\right)$, $p_{co,k}$ is simplified as 
\begin{align}
p_{co,k}&\approx 1 - \mathbb{E}_{L_k}\left[1 - \psi_k\beta_{t,k}^{\delta}\Xi_{M_k}\right]\nonumber\\
& = \mathbb{E}_{L_k}\left[\phi\lambda_o L_k^2\beta_{t,k}^{\delta}\Xi_{M_k}\right],
\end{align}
with $\Xi_{M_k}$ defined in Corollary \ref{corollary_cop_low}. Computing the above expectation by invoking \eqref{lk_pdf} gives the result in \eqref{cop_low}. 
 
\subsection{Proof of Proposition \ref{proposition_sop_general}}\label{proof_proposition_sop_general}
The SOP defined in \eqref{sop_def} can be rewritten cas 
\begin{align}\label{pso_appendix}
 &p_{so,k}=1- \mathbb{E}_{\Phi_e}\left[\prod_{e \in \Phi_{e}} \mathbb{P}\left\{{\rm{SINR}}_{e, k}<\beta_{e,k} | \Phi_{e}\right\}\right]\nonumber\\
&\stackrel{\mathrm{(a)}}=1- \exp\left(-\lambda_e
\int_0^{\infty}\int_0^{2\pi}\mathbb{P}
\left\{{\rm{SINR}}_{e, k}\ge\beta_{e,k}\right\}rd\theta dr\right),
\end{align}
where ${\rm{SINR}}_{e,k}$ is given by  \eqref{sinr_ek} with $r \triangleq r_{e,s_k} $, and 
(a) follows from the probability generating
functional (PGFL) over a PPP \cite{Chiu2013Stochastic}.

Defining $v\triangleq r^{\alpha}\beta_{e,k}/{P_a}$, then  $\mathbb{P}
\left\{{\rm{SINR}}_{e,k}
\ge\beta_{e,k}\right\}$ in \eqref{pso_appendix} can be calculated by invoking \cite[Eq. (11)]{Gao1998Theoretical}, i.e.,
\begin{equation}\label{pro_sinr}
\mathbb{P}\left\{{\rm{SINR}}_{e,k}
\ge\beta_{e,k}\right\}=
e^{-{\omega }{v}}
\sum_{m=1}^{M_e}\frac{(\omega v)^{m-1}}{(m-1)!}\mathbb{E}_{\Phi_j}\left[A_m(v)
\right],
\end{equation}
where $A_m(v) = \frac{\sum_{n=0}^{M_e-m}c_nv^n}{\prod_{ z\in \Phi_j}\left(1+P_jr_{ e,z}^{-\alpha}v\right)}$ with $c_n$ being the coefficient of $v^n$ in $\prod_{ z\in \Phi_j}\left(1+P_jr_{ e,z}^{-\alpha}v\right)$, which is
\begin{equation}\label{c_n}
c_n = \frac{1}{n!}
\sum_{	\mathcal{Z}_n	\subset\Phi_j}
\prod_{i=1}^{n}\frac{P_j}{r_{e,z_i}^{\alpha}},
\end{equation}
where $\mathcal{Z}_n\triangleq\{z_1,\cdots,z_n\}$ denotes an arbitrary subset of $n$ points selected from $\Phi_j$.
Substituting \eqref{c_n} into $A_m(v) $ yields
\begin{align}\label{pro_sinr2}
&\mathbb{E}_{\Phi_j}\left[A_m(v)
\right]=\sum_{n=0}^{M_e-m}
\frac{1}{n!} \underbrace{\mathbb{E}_{\Phi_j}\left[
 \sum_{\mathcal{Z}_{n}
	\in\Phi_j}\frac{P_j^{n}v^{n}
	\prod_{i=1}^{n}r_{e,z_i}^{-\alpha}}
{\prod_{z\in\Phi_j}\left(1+P_jr_{ e,z}^{-\alpha}v\right)}\right]}_{\mathcal{C}}.
\end{align}
Invoking Campbell-Mecke theorem \cite[Theorem 4.2]{Chiu2013Stochastic} with the term $\mathcal{C}$ in \eqref{pro_sinr2} gives
\begin{align}\label{campbell}
\mathcal{C}&=\left(2\pi\lambda_j\int_0^{\infty}
\frac{\vartheta}{1+\vartheta}tdt\right)^{n}
\exp\left(-2\pi\lambda_j\int_0^{\infty}
\frac{\vartheta}{1+\vartheta}tdt\right)\nonumber\\
&
\stackrel{\mathrm{(a)}} = \left(\phi\lambda_jP_j^{\delta}v
^{\delta}\right)^{n}\exp\left(
	-\phi\lambda_jP_j^{\delta}v
	^{\delta}\right),
\end{align}
where $\vartheta\triangleq P_j v t^{-\alpha}$, and (a) stems from \cite[Eq. (3.241.2)]{Gradshteyn2007Table}.
Substituting \eqref{pro_sinr2} with \eqref{campbell} into \eqref{pro_sinr} and plugging the result into  \eqref{pso_appendix}, the proof can be completed after some algebraic operations. 

\subsection{Proof of Proposition \ref{proposition_subopt_rho}}\label{proof_proposition_subopt_rho}
	We begin with re-expressing the objective function $T(\rho)$ in \eqref{sub_st_max} as $T(\rho)=\sum_{k=1}^K T_k(\rho)=\sum_{k=1}^K \ln\frac{W_k(\rho)}{Q(\rho)}$, where $W_k(\rho)=1+X_k(1+Y\rho)^{-\frac{\alpha}{2}}$ and $Q(\rho)=1+Z\rho^{-\frac{\alpha}{2}}$.
	For simplicity, we use the notations $T$, $T_k$, $W_k$, and $Q$, which are functions of $\rho$ by default. 
	The derivative $\frac{dT}{d\rho}$ can be calculated as 
	\begin{align}\label{dT}
	\frac{dT}{d\rho}&=\sum_{k=1}^K\frac{1}{W_k Q}\left(Q\frac{dW_k}{d\rho}-W_k\frac{dQ}{d\rho}\right)\nonumber\\
	&\stackrel{\mathrm{(a)}}=\sum_{k=1}^K\frac{1}{W_k Q}\left(-\frac{\alpha Y(W_k-1)Q}{2(1+Y\rho)}+\frac{\alpha(Q-1)W_k}{2\rho}\right)\nonumber\\
	&=\frac{\alpha}{2}\left(\frac{K(Q-1)}{\rho Q(1+Y\rho)}-\sum_{k=1}^K\frac{Y(W_k-Q)}{QW_k(1+Y\rho)}\right),
	\end{align}
	where $\mathrm{(a)}$ follows from the following two derivatives
	\begin{equation}\label{dWk}
	\frac{dW_k}{d\rho}=-\frac{\alpha}{2}X_k(1+Y\rho)^{-\alpha/2-1}Y=-\frac{\alpha}{2}\frac{Y(W_k-1)}{1+Y\rho},
	\end{equation}
	\begin{equation}\label{dQ}
	\frac{dQ}{d\rho}=-\frac{\alpha}{2}Z\rho^{-\alpha/2-1}=-\frac{\alpha(Q-1)}{2\rho}.
	\end{equation} 
	
	It is difficult to prove the concavity of $T$ w.r.t. $\rho$ by directly judging the monotonicity of $\frac{dT}{d\rho}$ from \eqref{dT}.
	In order to circumvent this issue, we reconstruct $\frac{dT}{d\rho}$ as below and resort to exploring the properties of its sign, 
	\begin{align}\label{dT2}
	\frac{dT}{d\rho}=\frac{\alpha K(Q-1)}{2\rho Q(1+Y\rho)}\left(\underbrace{1-\sum_{k=1}^K\frac{Y\rho(W_k-Q)}{KW_k(Q-1)}}_{G(\rho)}\right).
	\end{align}
	
	Clearly, the first term $\frac{\alpha K(Q-1)}{2\rho Q(1+Y\rho)}$ in \eqref{dT2} is constantly positive, and therefore the sign of $\frac{dT}{d\rho}$ is solely determined by that of the second term $G(\rho)$ which unfortunately is intuitively elusive. 
	To this end, we turn to examine the monotonicity of $G(\rho)$ on $\rho$ before determining its sign. Specifically, we rewrite $G(\rho)$ as $G(\rho) = 1 - \sum_{k=1}^K G_{1,k}(\rho) G_{2,k}(\rho)$ with $G_{1,k}(\rho)\triangleq \frac{Y\rho}{KW_k}$ and $G_{2,k}(\rho)\triangleq \frac{W_k-Q}{Q-1}$.
	Since both $W_k$ and $Q$ are decreasing functions of $\rho$, we can easily prove that $G_{1,k}(\rho)$ increases with $\rho$.
	Substituting $W_k$ and $Q$ into $G_{2,k}(\rho)$ yields $G_{2,k}(\rho) =\frac{X_k}{Z}\left({\rho}^{-1}+Y\right)^{-\alpha/2}-1$, which also increases with $\rho$.
	Consequently, $G(\rho)$ is a monotonically decreasing function of $\rho$.
	This result significantly facilitates the derivation of  the optimal $\rho^*$ that maximizes $T$ by  differentiating the following cases.
	
	1) Case $G(\rho_{\rm min})< 0$: Since $G(\rho)$ monotonically decreases with $\rho$, in this case $G(\rho)$ or $\frac{dT}{d\rho}$ keeps negative within $\rho\in[\rho_{\rm min}, 1]$. 
	This indicates that $T$ is a decreasing function of $\rho$, and the optimal $\rho^*$ reaching the maximal $T$ is $\rho^*=\rho_{\rm min}$.
	The condition $G(\rho_{\rm min})< 0$ can be equivalently transformed as below by plugging $\rho=\rho_{\rm min}$ into $G(\rho)$,
	\begin{align}\label{case1}
	G(\rho_{\rm min}) &= 1 - \sum_{k=1}^K \frac{Y\rho_{\rm min}}{KZ}\frac{X_k\left(Y+1/\rho_{\rm min}\right)^{-\alpha/2}-Z}{X_k\left(1+Y\rho_{\rm min}\right)^{-\alpha/2}+1}<0\nonumber\\
&	\Leftrightarrow Z < 
	\frac{Y\rho_{\rm min}^{\alpha/2+1}\left(1-\kappa(\rho_{\rm min})\right)}{1+ Y\rho_{\rm min}\kappa(\rho_{\rm min})}.
	\end{align}
	
	2) Case $G(1)\ge 0 $: Again, due to the monotonically-decreasing feature of $G(\rho)$ on $\rho$, in this case $G(\rho)$ or $\frac{dT}{d\rho}$ maintains positive within $\rho\in[\rho_{\rm min}, 1]$. 
	In other words, $T$ monotonically increases with $\rho$ and is maximized at $\rho=1$. 
	Following \eqref{case1}, the condition $G(1)\ge 0$ is equivalent to $Z\ge \frac{Y(1-\kappa(1))}{1+Y\kappa(1)}$.
	
	3) Case $G(1)< 0 \le G(\rho_{\rm min})$: In this case, as $\rho$ increases from $\rho_{\rm min}$ to $1$, $G(\rho)$ or $\frac{dT}{d\rho}$ is initially positive and then becomes negative, which implies $T$ first increases and then decreases with an increasing $\rho$, and there exists a unique peak value of $T$. 
	Obviously, the maximal $T$ is obtained when $\rho$ arrives at the zero-crossing point of $G(\rho)$ or $\frac{dT}{d\rho}$.
	
	By now, the proof is completed.

\subsection{Proof of Corollary \ref{corollary_opt_rho}}\label{proof_corollary_opt_rho}
The results provided by Corollary \ref{corollary_opt_rho} can be obtained by examining the derivatives of $\rho^{\circ}$ on the variables $X_k$, $Y$, and $Z$, respectively, invoking the derivative rule for implicit functions with the equation $G(\rho^{\circ})=0$ with $G(\rho)$ defined in \eqref{G} \cite{Zheng2015Multi}.

To begin with, the derivative $\frac{d \rho^{\circ}}{d X_k}$ can be calculated as 
 \begin{equation}\label{dX}
 \frac{d \rho^{\circ}}{d X_k}=-\frac{\partial G(\rho^{\circ})/\partial X_k}{\partial G(\rho^{\circ})/\partial \rho^{\circ}}.
 \end{equation}
It is easy to see that $\frac{\partial G(\rho^{\circ})}{\partial \rho^{\circ}}<0$ since $G(\rho)$ is a decreasing function of $\rho$.
Besides, as $\kappa(\rho^{\circ})$ is a decreasing function of $X_k$ and at the same time $G(\rho^{\circ})$ increases with $\kappa(\rho^{\circ})$, we have $\frac{\partial G(\rho^{\circ})}{\partial X_k}<0$, which yields $\frac{d \rho^{\circ}}{d X_k}<0$.
Similarly, we can readily show $\frac{d \rho^{\circ}}{d Z}>0$.
We can also prove that $\frac{d \rho^{\circ}}{d Y}<0$ by noting that 
\begin{align}
\frac{\partial G(\rho^{\circ})}{\partial Y}  = \rho^{\circ} \kappa(\rho^{\circ})-\frac{\left(\rho^{\circ}\right)^{\alpha/2+1}\left(1-\kappa(\rho^{\circ})\right)}{Z}\stackrel{\mathrm{(a)}}= -\frac{1}{Y}<0,
\end{align}
where $\mathrm{(a)}$ follows from $G(\rho^{\circ})=0$.
After obtaining the above results, the proof can be completed by simply determining the relationships between the key parameters $\{\sigma,\epsilon,M_c,\lambda_s,P_j/P_a,\lambda_c,K,\lambda_e,M_e\}$ and the auxiliary variables $\{X_k,Y,Z\}$. 
Due to space limitation, we just take the SOP threshold $\epsilon$ as an example, where we observe that $Z$ decreases with $\epsilon$ and meanwhile $\frac{d \rho^{\circ}}{d Z}>0$ such that $\frac{d \rho^{\circ}}{d \epsilon}<0$.

\end{document}